\pdfoutput=1
\documentclass[11pt]{amsart}
\usepackage{amsmath,amsxtra,amssymb,amsthm,amsfonts}

\usepackage{bookmark}
\usepackage{mathrsfs}
\usepackage[T1]{fontenc}
\usepackage[utf8]{inputenc}
\usepackage{mathtools}   % loads »amsmath«
\usepackage{color}
\usepackage{cite}
\numberwithin{equation}{section}
\newtheorem{lemma}{Lemma}[section]
\newtheorem{prop}[lemma]{Proposition}
\newtheorem{theorem}[lemma]{Theorem}
\newtheorem{cor}[lemma]{Corollary}

\newtheorem{prob}[lemma]{Problem}
\newtheorem{rem}[lemma]{Remark}
\newtheorem{remark}[lemma]{Remark}

\newtheorem{example}[lemma]{Example}
\newtheorem{definition}[lemma]{Definition}
\newtheorem{corollary}[lemma]{Corollary}

\newcommand{\re}{\begin{rem}\rm}
\newcommand{\mar}{\end{rem}}

\newcommand{\ee }{\mathrm{I}\!\!1}

\renewcommand{\for}{\begin{eqnarray*}}

\newcommand{\mel}{\end{eqnarray*}}

\newcommand{\kl}{\pl \le \pl}
\newcommand{\gl}{\pl \ge \pl}

\newcommand{\lel}{\pl = \pl}

\newcommand{\ez}{{\mathbb E}}

\newcommand{\rz}{{\mathbb R}}

\newcommand{\Mz}{{\mathbb M}}

\newcommand{\cz}{{\mathbb C}}

\newcommand{\ten}{\otimes}

\newcommand{\pl}{\hspace{.1cm}}
\newcommand{\pll}{\hspace{.3cm}}

\newcommand{\qd}{\end{proof}\vspace{0.5ex}}

\newcommand{\om}{\omega}

\newcommand{\al}{\alpha}

\newcommand{\si}{\sigma}

\newcommand{\la}{\lambda}
\newcommand{\eps}{\varepsilon}

\newcommand{\id}{\iota_{\infty,2}^n}

\renewcommand{\S}{{\mathcal S}}

\newcommand{\pf}{\begin{proof}}

\newcommand{\be}{\left|{\atop}}

\newcommand{\summ}{\sum\limits}

\newcommand{\xspace}{\hbox{\kern-2.5pt}}
\newcommand{\xyspace}{\hbox{\kern-1.1pt}}

\newcommand\bra[1]{\langle  #1|}
\newcommand\ket[1]{| #1\rangle}

\allowdisplaybreaks

\definecolor{LightGray}{rgb}{0.94,0.94,0.94}
\definecolor{VeryLightBlue}{rgb}{0.9,0.9,1}
\definecolor{LightBlue}{rgb}{0.8,0.8,1}
\definecolor{DarkBlue}{rgb}{0,0,0.6}
\definecolor{LightGreen}{rgb}{0.88,1,0.88}
\definecolor{MidGreen}{rgb}{0.6,1,0.6}
\definecolor{DarkGreen}{rgb}{0,0.6,0}
\definecolor{DarkGrreen}{rgb}{0,0.8,0}

\definecolor{VeryLightYellow}{rgb}{1,1,0.9}
\definecolor{LightYellow}{rgb}{1,1,0.6}
\definecolor{MidYellow}{rgb}{1,1,0.5}
\definecolor{DarkYellow}{rgb}{0.8,1,0.3}
\definecolor{VeryLightRed}{rgb}{1,0.9,0.9}
\definecolor{LightRed}{rgb}{1,0.8,0.8}
\definecolor{DarkRed}{rgb}{0.8,0.2,0}
\definecolor{DarkRedb}{rgb}{0.6,0.2,0}
\definecolor{DarkLila}{rgb}{0.8,0,1}
\definecolor{Beige}{rgb}{0.96,0.96,0.86}
\definecolor{Gold}{rgb}{1.,0.84,0.}
\definecolor{Goldb}{rgb}{0.7,0.3,0.5}
\definecolor{MyYellow}{rgb}{1.,0.84,0.8}

\newcommand{\lan}{\langle}
\newcommand{\ran}{\rangle}

\def\11{\mathbb{I}}

\usepackage{graphicx}

\DeclareRobustCommand\openone{\leavevmode\hbox{\small1\normalsize\kern-.33em1}}

\renewcommand{\id}{\rm{id}}
\renewcommand{\be}{\begin{equation}}
	\renewcommand{\ee}{\end{equation}}
\newcommand{\bea}{\begin{eqnarray}}
	\newcommand{\eea}{\end{eqnarray}}
\newcommand{\beas}{\begin{eqnarray*}}
	\newcommand{\eeas}{\end{eqnarray*}}

\usepackage{amsmath}
\usepackage{ams math}
\usepackage{amssymb}
\usepackage{amsthm}
\usepackage{dsfont}
\usepackage{upgreek}
\usepackage[shortlabels]{enumitem}
\usepackage{array}
\usepackage{ams fonts}
\usepackage{graphics}
\usepackage[utf8]{inputenc}
\newtheorem*{theorem*}{Theorem}
\newtheorem*{remark*}{Remark}
\newtheorem*{lemma*}{Lemma}
\newtheorem*{note*}{Note}
\newtheorem*{prop*}{Proposition}

\newcommand{\Tr}{\mbox{Tr}}

\usepackage{geometry}
\usepackage{tikz-cd}
\newcommand{\ti}{\ket{\tilde{i}}}
\newcommand{\trho}{\Tilde{\rho}}
\newcommand{\op}[2]{|#1\rangle\langle #2|}

\newcommand{\wt}{\widetilde}
\newcommand{\q}{\quad}

\begin{document}

\title{Quantum Secret Sharing and Tripartite Information}
\author[M. Junge]{Marius Junge}
\address{Department of Mathematics\\
University of Illinois, Urbana, IL 61801, USA}
\email[Marius Junge]{junge@math.uiuc.edu}
\author[C. Kim]{Chloe Kim}
\address{Department of Electrical \& Computer Engineering\\
University of Illinois, Urbana, IL 61801, USA}
\email[Chloe Kim]{kim705@illinois.edu}
\author[G. Liu]{Guangkuo  Liu}
\address{Department of Physics\\
University of Illinois, Urbana, IL 61801, USA}
\email[Guangkuo Liu]{gl5@illinois.edu}
\author[P. Wu]{Peixue Wu}
\address{Department of Mathematics\\
University of Illinois, Urbana, IL 61801, USA}
\email[Peixue Wu]{peixuew2@illinois.edu}

\thanks{MJ is partially supported by NSF grants  DMS 1800872 and Raise-TAG 1839177. }

\maketitle
\begin{abstract}
We develop a connection between tripartite information $I_3$, secret sharing protocols and multi-unitaries. This leads to explicit $((2,3))$ threshold  schemes in arbitrary dimension minimizing tripartite information $I_3$. As an application we show that Page scrambling unitaries simultaneously work for all secrets shared by Alice. Using the $I_3$-Ansatz for imperfect sharing schemes we discover examples  of VIP sharing  schemes.

%Based on the connection of tripartite information and perfect secret sharing protocol, we talk about some extensions. First we give an example where only subset of parties can recover the secret but not perfect for all parties. Moreover, for general secret sharing protocal, which may not be perfect for any parties, we have the estimate of tripartite information, given by the operator norm of the tensor map. As an application, for random unitaries, we give the random estimates of $I_3$, with a dimension free upper bound for any the densities in the secret sharing protocol.
\end{abstract}

\section{Introduction}
A major resource for many quantum tasks are EPR-pairs. In contrast, this note quantifies three party entanglement using different concepts in quantum information theory. Three party entanglement is even more complex than two party entanglement. A major obstacle in understanding three party entanglement is the absence of a suitable analogue of entanglement entropy. The analogue of the maximally entangled state, the GHZ state, is insufficient in producing unbounded multipliciative violations of high dimensional Bell inequalities \cite{JPPV}.
Tripartite information was first introduced as ``topological entropy" in \cite{Kitaev2006} to characterize multi-party entanglement in a topologically ordered system. Tripartite information $I_3$  and out-of-order-correlations are used to measure delocalization of information in the bulk-boundary picture of the AdS/CFT correspondence, see e.g. \cite{HayWal}, and see \cite{Net} for a connection to  neural networks, and e.g. \cite{Schna} for further information.

%In this paper we establish a basic equivalence between secret sharing protocols and tripartite information.

%For once three multiparty entanglement occurs naturally in the entropic interpretation of holographic states. On the other hand, in contrast to two partite correlations scenario,   multiplicative violations of Bell-type inequalities for three party correlations become potentially  unbounded by increasing the dimension, see e.g. \cite{JPP}. The real obstacle in understanding even three party entanglement is the absence of suitable analogue of entanglement entropy, as a qualitative  measure of entanglement as a resource. Moreover, the famous GHZ state, the natural analogue of a maximally entangled state is no longer the only relevant resource for three party entanglement.

In this paper we will develop a connection between tripartite information $I_3$ and secret sharing protocols. Tripartite information is defined for any tripartite state as
 \begin{align*}
 I_3(P_1:P_2:P_3) &= S(P_1)+S(P_2)+S(P_3)-S(P_1P_2)-S(P_1P_3)-S(P_2P_3)+S(P_1P_2P_3)
 \end{align*}
 and enjoys many symmetry properties. As pointed out in \cite{Yoshida}, the notion $I_3$ measures indeed entanglement for the four partite pure state  $\ket{\phi}^{RP_1P_2P_3}$ given by the purification of the density matrix $\rho^{P_1P_2P_3}$. We refer to \cite{Yoshida} for the rich symmetry properties of $I_3$. In particular, we recall that
 \begin{align*}\label{pp}
 I_3(P_1:P_2:P_3) &=  I(R,P_1)+I(R,P_2)-I(R,P_1P_2)\\
 &= -2S(R)+ I(R,P_1)+I(R,P_2)+I(R,P_3) \pll .\end{align*}
As it is the case for conditional entropy,
$I_3$  may have positive and negative values. A negative or even strictly negative value is an indication of existing entanglement. In \cite{HayWal} it is shown that under the premisses of famous Hayden-Preskill Gedankenexperiment $I_3$ is always strictly negative. The setting in this paper, following \cite{Yoshida} is more general.
In \cite{Ha1} Harlow analyzes the situation of two finite dimensional von Neumann algebras $M,M'\subset \mathbb{B}(H_{code})$ which are sent to two different registers
 \[ H_{code}\subset H\bar{\ten} H' \]
such that
\begin{align} \label{Har}
 D(\rho^M||\si^{M})\lel D(\rho^H||\si^H)
\quad \mbox{and} \quad D(\rho^{M'}||\si^{M'})\lel D(\rho^{H'}||\rho^{H'})
 \end{align}
holds for all densities in a code space. Harlow shows the equivalence of simultaneous relative entropy recovery and the existence of a certain pair of unitaries, and a state measuring entanglement threads.

In quantum information theory the equality case in entropic inequalities often occurs under specific algebraic requirements. Therefore it is natural to ask for equality in the obvious lower bound
    \[ -2S(R) \kl   I_3(P_1:P_2:P_3)  \pl . \]
The bound follows easily from the positivity  of mutual information, see e.g. \cite{Wilde}.
 As pointed out by \cite{Yoshida}, see also \cite{preskill}, random unitaries and perfect tensor almost achieve equality of $I_3$ in many cases, when applied to maximally entangled states between $M$ and $M'$. Let us recall the famous Hayden-Preskill Gedankenexperiment where Alice's secret can be recovered from the gamma radiation of a black hole which has been observed for a long time.  This leads to an interpretation of scrambling in terms of error correction and decoding, closely connected to the powerful tool of decoupling, \cite{preskill}.

The link to a secret sharing protocol follows from Harlow's setup considering
a subspace
 \[ \S\subset P_1P_2P_3 \pll .\]
By fixing a basis $|\tilde{i}\ran$ the Referee may send a signal from his register $R$ of the same dimension $r$ by constructing the purification
%$|\psi_{\la}\ran$
$$\ket{\phi}^{RP_1P_2P_3}=\sum_{i=1}^r\sqrt{\lambda_i}\ket{\psi_i}^R\ket{\wt{\psi_i}}^{P_1P_2P_3}$$
of any density matrix $\trho=\sum_{i}\lambda_i\op{\wt{\psi_i}}{\wt{\psi_i}}$ in $\S$.

We will say that fixing a basis $\{|\tilde{i}\ran\}_{i=1}^r$  defines a \emph{quantum secret sharing scheme}, if a referee can send a secret state to a collection of untrusted parties such that only authorized subsets of the parties can reconstruct the secret perfectly while those that are unauthorized gain zero information. Quantum secret sharing has been a long-established topic in quantum information theory \cite{Lo}. It has a variety of applications such as quantum money and quantum resource distribution, to name a few \cite{Lo}. Quantum secret
sharing protocols can also be defined in terms of the mutual information among subsets of the given collection of parties and the ancillary party \cite{Bai}. In our situation minimality of $I_3$ is equivalent to so-called $((2,3))$ threshold schemes, i.e. the secret is sent to three parties, and all pair of them can perfectly recover the secret, but no single party can. In the following text, we will use the terms $((2,3))$ threshold scheme and secret sharing scheme interchangeably without ambiguity.

We observe that there is a one-to-one correspondence between minimality of $I_3$ and the $((2,3))$ secret sharing schemes. Moreover, if the code space is fixed with the same dimension:  \[ d\lel |S|\lel |P_1| \lel |P_2| \lel |P_3| \pl ,\]
then they admit a third, equivalent formulation using multi-unitaries:  \\
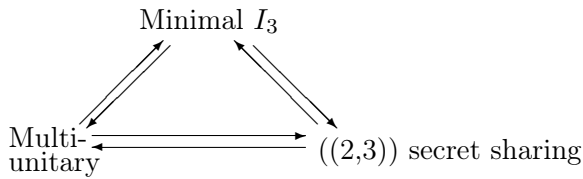
\begin{center}

  \setlength{\unitlength}{0.9pt}
  \begin{picture}(260,100)
  \put(80,100){Minimal $I_3$} \put(25,50){Multi-} \put(155,45){((2,3)) secret sharing}
  \put(25,40){unitary}
  \put(55,60){\vector(1,1){35}} \put(93,93){\vector(-1,-1){35}}
  \put(155,60){\vector(-1,1){35}} \put(128,93){\vector(1,-1){35}}
  \put(60,55){\vector(1,0){90}} \put(150,50){\vector(-1,0){90}}
  \end{picture}
\end{center}
%In this setting a secret sharing scheme is equivalent to
%  \[ I(R,P_a) \lel 0 \]
%for all $a$ and all inputs. Moreover, the output state  is necessarily $\rho^{P_a}$ maximally mixed for all three parties. Then the multi-unitary is the four leg  tensor tensor $t$ defined by $S$ and the choice of the basis
% \[ \tilde{i} \lel d^{-1/2} t_{i,s_1,s_2,s_3} |s_1s_2s_2\ran  \pl .\]
%The multi-unitary condition requires $t^{RP_a,P_bP_c}$ to be unitary for all three choices of $a$. We refer to section ... for examples.

If we choose a basis, the four leg tensor $t$ is given by
 \begin{equation}\label{tensorcon}
|\tilde{i}\ran \lel \sum_{s_1,s_2,s_3} \frac{1}{\sqrt{d}} t_{is_1s_2s_3} |s_1s_2s_3\ran \pl.
\end{equation}
 The multi-unitary condition then refers to the condition that for all three choices $a=1,2,3$, we have that $t:=\sum_{i,s_1,s_2,s_3}t_{is_1s_2s_3}\op{s_bs_c}{is_a}$ becomes a unitary. We refer to section 2 for details.

Using the predicted relation between sharing schemes and $I_3$ one can also produce new examples with small, but not necessarily minimal values, of $I_3$, as in Page scrambling.  Indeed, for a unitary  $u:RP_1\to P_2P_3$, and we may define the tensor $t_{i,s_1,s_2,s_3}=u_{(is_1),(s_2s_3)}$ and then estimate tripartite information:

\begin{theorem} Let $u:RP_1\to P_2P_2$ be a random unitary. Then with probability $1-\delta$
 \[ -2S(R)\kl I_3(P_1P_2P_3)\kl -2S(R)+C(\delta) \]
holds for \emph{all} purified input states $\ket{\phi}$.
\end{theorem}

Previous result were usually restricted to a maximally entangled state $\ket{\phi}$ and did not provide good enough concentration of measure to work for all densities simultaneously. Although this  estimate is not as tight as Page's original estimate it adds a concrete dimension free relation between $\delta$ and $C(\delta)$. As observed by Kitaeev an additive factor is expected as long as $|P_2|=|P_3|$.

The paper is organized as follows. After some preliminaries in section 1, the equivalence of the three conditions and the connection to Uhlmann's theorem is presented in section 2. In section 3, we study conditions satisfying the minimality of $I_3$ up to a constant term, and in section 4, this is shown to be generic for random matrices provided $|S|=|P_j|$. Section 5 provides concrete examples of perfect secret sharing schemes that work for all dimension, and also an example of imperfect secret sharing schemes that always requires a fixed party to be present to recover a secret.

\section{Equivalence Conditions for Perfect Secret Sharing Schemes}
As mentioned in the introduction, the ((2,3)) threshold scheme means that a referee sends a secret to three parties, and any two parties can recover the secret; however, any single party alone is forbidden from decoding the secret. This ability to recover from the erasure of one party has been well studied in the context of quantum error correction \cite{Schumacher}. In fact, Harlow et al.\cite{Harlow} has shown, using Uhlmann's theorem, that for a $\S\subset P_1P_2P_3$ spanned by $\{\ti\}_{i=1}^{r}$, and for any party $P_a$ to be erased, the following are equivalent:
\begin{itemize}
    \item[(i)]  $I(R, P_a)=0$ for the state $\ket{\phi}=1/\sqrt{r}\sum_{i}\ket{i}^R\ket{\wt{i}}^{P_1P_2P_3}$,
    \item[(ii)] there exists a unitary $U_{a}:RP_a'\rightarrow P_bP_c$ such that for any $i=1,2,\dots,r$, \begin{align}\label{eqn:uni}
    \ti^{P_aP_bP_c}
    =\id^{P_a}\otimes U_{a}\ket{i}^R\ket{\chi_a}^{P_a'P_a},
\end{align}
where $\ket{\chi_a}^{P_a'P_a}$ does not depend on the choice of $i$, and the prime denotes an ancillary copy of the party.
\end{itemize}
We should make a remark that (ii) is the condition for two parties $P_bP_c$ to be able to recover any state sent by referee $R$, using unitary transformations acting only on their parties.  Indeed, by applying the unitary mapping $\id^{P_a}\otimes U_{a}^{\dagger P_bP_c\rightarrow RP_a'}$ on $P_bP_c$, we can actually recover any density matrix in $\S$, i.e., a matrix of the form $\trho^{P_aP_bP_c}=\sum_{ij}\rho_{ij}\op{\wt{i}}{\wt{j}}$ will be mapped to $\rho^R\otimes \chi_a^{P_a'P_a}=\sum_{ij}\rho_{ij}\op{i}{j}^R\otimes\op{\chi_a}{\chi_a}^{P_a'P_a}$. Thus the message  $\rho^R$ sent by the referee is recovered.
%Conversely, following \cite{Harlow}, we may apply Uhlmann's theorem to (i), and find a unitary $U_a:RP_a'\to P_bP_c$ such that
 %\begin{align*}
%  \id\ten U_a\left(
%  \sum_{is} \sqrt{\la_s\mu_s} |is\ran\ten |is\ran^{RP_a'}
%  \right)
% &=  \frac{1}{\sqrt{r}} \sum_{i=1}^r |i\ran \ten |\tilde{i}\ran\\
% &=  \sum_{i,s_1} \sqrt{\frac{\la_s}{r}}
% |is_1\ran\ten |\phi_{is}\ran^{P_bP_c} \pl .
% \end{align*}
%Comparing coefficients we deduce that the  map $U_a(|is\ran)\lel |\phi_{is}\ran$ defines the unitary in (ii).
\par
Since the secret sharing scheme requires the recoverability against the erasure of any party $P_1$, $P_2$ or $P_3$, we have $$I(R, P_1)=I(R,P_2)=I(R,P_3)=0.$$
Recall that this is exactly the necessary and sufficient condition to have minimal $I_3$:\begin{align*}
    I_3&=-2S(R)+I(R, P_1)+I(R,P_2)+I(R,P_3)\\
    &\geq -2S(R),
\end{align*}
with equality obtain if and only if $I(R, P_j)=0$ for all $j=1,2,3$. Therefore, we obtain two equivalent definitions of a ((2,3)) threshold scheme: \begin{definition}\label{def:secretsharing} A code space $\S\subset P_1P_2P_3$ spanned by $\{\ti\}_{i=1}^{r}$ is a ((2,3)) threshold scheme if either
\begin{itemize}
    \item[{ (a)}] $I_3=-2S(R)$ for the state $\ket{\phi}=1/\sqrt{r}\sum_{i}\ket{i}^R\ket{\wt{i}}^{P_1P_2P_3}$,
    \item[{ or (b)}] There exist three unitary maps, $U_{a}:RP_a'\rightarrow P_bP_c$, $a=1,2,3$ that all satisfy (\ref{eqn:uni}).
\end{itemize}
\end{definition}
In contrast to the traditional definition of secret sharing schemes \cite{Bai}, which only uses mutual information, we define it using $I_3$ because it quantifies imperfection for sharing schemes. See section 3 for details.\par
We now discuss another equivalent definition for a perfect secret sharing scheme in the special case when $|R|=|P_1|=|P_2|=|P_1|=d$. Let us now define the secret subspace $\S\subset P_1P_2P_3$ by fixing its basis as \begin{align}\label{eqn:tensor-rep}
    \ket{\wt{i}}=\sum_{s_1,s_2,s_3=1}^{d} \frac{1}{\sqrt{d}} t_{is_1s_2s_3} \ket{s_1s_2s_3}^{P_1P_2P_3}\text{, }i=1,2,\dots,d.
\end{align}
\begin{theorem}\label{lemma:m-u}
$\S$ is a ((2,3)) threshold scheme if and only if $t_{is_1s_2s_3}$ is multi-unitary \cite{multiunitary}:\begin{itemize}
    \item[1.] the map $t:=\sum_{i,s_1,s_2,s_3}t_{is_1s_2s_3}\op{s_2s_3}{is_1}$ is unitary,
    \item[2.] its reshuffling, $t^R:=\sum_{i,s_1,s_2,s_3}t_{is_1s_2s_3}\op{s_1s_3}{is_2}$ is unitary,  and
\item[3.] its partial transposition (followed by a flip), $t^{\Gamma}:=\sum_{i,s_1,s_2,s_3}t_{is_1s_2s_3}\op{s_1s_2}{is_3}$ is unitary.
\end{itemize}
Moreover, the error-correcting unitaries $U_1$, $U_2$ and $U_3$ are uniquely given, up to a local unitary, by \begin{align}\label{eqn:get-u}
\begin{aligned}
&U_1=\sum_{i,s_1,s_2,s_3}t_{is_1s_2s_3}\ket{s_2s_3}^{P_2P_3}\bra{is_1}^{RP_1'},\\
    &U_2=\sum_{i,s_1,s_2,s_3}t_{is_1s_2s_3}\ket{s_1s_3}^{P_1P_3}\bra{is_2}^{RP_2'}\text{, and}\\
    &U_3=\sum_{i,s_1,s_2,s_3}t_{is_1s_2s_3}\ket{s_1s_2}^{P_1P_2}\bra{is_3}^{RP_3'}.
\end{aligned}
\end{align}
\end{theorem}
Before proving Theorem \ref{lemma:m-u}, we first show some properties of ((2,3)) sharing schemes.
\begin{prop}\label{prop:alli3}
Let $\S$ be a ((2,3)) threshold scheme, then for any $\trho^{P_1P_2P_3}
%=\sum_{ij}\rho_{ij}\op{\tilde{i}}{\tilde{j}}
\in \S$, its reduced densities $\trho^{P_a}$ is independent of the choice of $\trho^{P_1P_2P_3}$, and for any $\trho^{P_1P_2P_3}\in\S$, $I_3=-2S(R)$.
\end{prop}
\begin{proof}
Suppose $\trho=\sum_{ij}\rho_{ij}\op{\tilde{i}}{\tilde{j}}$. Let $P_a$ denote any party, and $P_b$, $P_c$ the remaining parties. By definition (b), for some unitary $U_a:RP_a'\rightarrow P_bP_c$\begin{align}
    &\trho^{P_aP_bP_c}=
    U_a\otimes\id^{P_a}\left(
    \sum_{ij}\rho_{ij}\op{i}{j}^{R}\otimes
    \op{\chi_a}{\chi_a}^{P_a'P_a}
    \right)
    U_a^\dagger \otimes\id^{P_a}.
\end{align}
If we take the partial trace over $P_b$ and $P_c$, we get
\begin{align}
     \trho^{P_a}=\Tr_{P_bP_c}(\trho^{P_aP_bP_c})=\Tr_{P_a'}\left(\op{\chi_a}{\chi_a}^{P_a'P_a} \right),
\end{align}
which is independent of the choice of $\trho$. If we take partial trace over $P_a$, we get
\begin{align}
    \trho^{P_bP_c}&=\Tr_{P_a}(\trho^{P_aP_bP_c})=
    U_a\left(
    \sum_{ij}\rho_{ij}\op{i}{j}^{R}\otimes
    \Tr_{P_a}\left(\op{\chi_a}{\chi_a}^{P_a'P_a}\right)
    \right)U_a^\dagger .
\end{align}
Since von Neumann entropy is invariant under unitary, we have \begin{align}
    S(P_bP_c)=S\left(\sum_{ij}\rho_{ij}\op{i}{j}\right)+S\left(\Tr_{P_a}\op{\chi_a}{\chi_a}^{P_a'P_a}\right)=S(R)+S(P_a).
\end{align}
Substituting this into the expression of $I_3$, we get $I_3=-2S(R)$.
\end{proof}
\begin{corollary}\label{cor:one->all}
$I_3=-2S(R)$ for the state $\ket{\phi}=1/\sqrt{r}\sum_{i}\ket{i}^R\ket{\wt{i}}^{P_1P_2P_3}$ implies $I_3=-2S(R)$ for all $\trho\in\S$.
\end{corollary}
The above lemma does not assume the parties to have the same dimension. If we add this assumption, we have
\begin{prop}
Let $\S$ be a ((2,3)) threshold scheme and assume that $|R|=|P_j|=d$ for $j=1,2,3$, then for any $\trho^{P_1P_2P_3}\in \S$, its reduced densities for any party must be maximally mixed, ie. for any $P_a$, $\trho^{P_a}=\Tr_{P_bP_c}\trho^{P_aP_bP_c}=\sum_s \frac{1}{d}\op{s}{s}^{P_a}$.
\end{prop}
\begin{proof}
From Lemma \ref{prop:alli3}, it is sufficient to show for only $\trho^{P_1P_2P_3}=\sum_i\frac{1}{d}\op{\tilde{i}}{\tilde{i}}$, since the reduced density is independent of the choice of $\trho^{P_1P_2P_3}$. \par
For this state, $I_3=-2S(R)=-2\log d$. From the fact that $I_3$ is symmetric with respect to the choice of parties, we have $I_3\geq-2S(P_a)$ for any party $P_a$. Therefore, $S(P_a)\geq \log d$, which is possible only when $\trho^{P_a}$ is maximally mixed.
\end{proof}
\begin{remark}
Since $\ket{\chi_a}^{P_a'P_a}$ is a purification of $\trho^{P_a}=\sum_s\frac{1}{d}\op{s}{s}^{P_a}$, up to a local unitary, it is of the form $\ket{\chi_a}=\sum_s \frac{1}{\sqrt{d}}\ket{s}^{P_a'}\ket{s}^{P_a}$. Incorporating the local unitary into $U_a$, without loss of generality, we can restate equation (\ref{eqn:uni}) as \begin{align}\label{eqn:new-uni}
    \ti^{P_aP_bP_c}=U_a
    %^{RP_a'\rightarrow P_bP_c}
    \otimes\id^{P_a}\sum_s\frac{1}{\sqrt{d}}\ket{i}^{R}\ket{s}^{P_a'}\ket{s}^{P_a}.
\end{align}
\end{remark}
With these properties, we can prove Theorem \ref{lemma:m-u}.
\begin{proof}[Proof of Theorem  \ref{lemma:m-u}.]
Let $\S\subset P_1P_2P_3$ be a ((2,3)) sharing scheme spanned by $\{\ti\}_{i=1}^{d}$. According to definition (b), we have three unitaries $U_1$, $U_2$, $U_3$ that all satisfy (\ref{eqn:new-uni}). \par
Denote
$t_{is_1s_2s_3}$ to be the coefficients of $U_1$, i.e. \begin{align}
t_{is_1s_2s_3}=\bra{s_2s_3}^{P_2P_3}U_1\ket{is_1}^{RP_1'} .
\end{align}
By substituting this into (\ref{eqn:new-uni}), we already have a tensor representation of the basis as in (\ref{eqn:tensor-rep}),\begin{align}
    \ket{\wt{i}}=\sum_{s_1,s_2,s_3=1}^{d} \frac{1}{\sqrt{d}} t_{is_1s_2s_3} \ket{s_1s_2s_3}^{P_1P_2P_3},\nonumber
\end{align}
and that $$U_1=\sum_{i,s_1,s_2,s_3}t_{is_1s_2s_3}\ket{s_2s_3}^{P_2P_3}\bra{is_1}^{RP_1'}.$$
Our goal is to show that $t_{is_1s_2s_3}$ must be multi-unitary. We already have the condition 1 that $t:=\sum t_{is_1s_2s_3}\op{s_2s_3}{is_1}$ is unitary. \par
Assume $U_2=\sum_{j,k,l,m}g_{jk,lm}\ket{lm}^{P_1P_3}\bra{jk}^{RP_2'}$, then apply (\ref{eqn:new-uni}) twice, we get
\begin{align}\label{eqn:bigugly}
\sum_s\frac{1}{\sqrt{d}}&\ket{i}^{R} \ket{s_2}^{P_2'} \ket{s_2}^{P_2}
    =\left(U_2^{\dagger P_1P_3\rightarrow RP_2'}\otimes \id^{P_2}\right)\ti \nonumber\\
    =&\sum_{s_1}\frac{1}{\sqrt{d}}\left(U_2^{\dagger P_1P_3\rightarrow RP_2'}\otimes \id^{P_2}\right)\left(U_1^{ RP_1'\rightarrow P_2P_3}\otimes \id^{P_1}\right)\ket{i}^{R} \ket{s_1}^{P_1'} \ket{s_1}^{P_1}\nonumber \\
    =&\sum_{k,s_2}\frac{1}{\sqrt{d}}\left(
    \sum_{s_1s_3}g_{jk,s_1s_3}^*t_{is_1 s_2s_3}
    \right)
    \ket{j}^R\ket{k}^{P_2'}\ket{s_2}^{P_2}.
\end{align}
So one must have $\sum_{s_1s_3}g_{jk,s_1s_3}^*t_{is_1 s_2s_3}=\delta_{ij}\delta_{k,s_2}$. This is same as saying the maps \begin{align*}
    &g:=\sum g_{jk,lm}\op{lm}{jk}\\
    &t^R:=\sum t_{is_1s_2s_3}\op{s_1s_3}{is_2}
\end{align*}
satisfy $g^\dagger t^R=\id$. Since $g$ is a unitary, we have that $g=t^R$,  So the reshuffling $t^R$ must also be unitary, i.e. the condition 2 of multi-unitary is satisfied. In addition, we have $$U_2=\sum_{i,s_1,s_2,s_3}t_{is_1s_2s_3}\ket{s_1s_3}^{P_1P_3}\bra{is_2}^{RP_2'}$$
By repeating the same procedure for $U_3$ we can see that the partial transpose $t^{\Gamma}:=\sum t_{is_1s_2s_3}\op{s_1s_2}{is_3}$ must also be unitary and it gives the coefficient expression of $U_3$:
$$U_3=\sum_{i,s_1,s_2,s_3}t_{is_1s_2s_3}\ket{s_1s_2}^{P_1P_2}\bra{is_3}^{RP_3'}.$$
\end{proof}
We will exploit this algebraic form in the next section.

\section{Small \texorpdfstring{$I_3$}{}}

In this section we investigate code spaces which almost achieve the ((2,3)) threshold schemes. As usual we start with a code space $\S\subset P_1P_2 P_3$ and a fixed orthonormal basis
 \[ |\tilde{i}\ran^{P_1P_2P_3}  \lel \sum_{s_1,s_2,s_3}\frac{1}{\sqrt{d}} t_{is_1s_2s_3} |s_1s_2s_3\ran
   \pll .\]
  However, now $t_{is_1s_2s_3}$ does not need to be multi-unitary. We can still define the linear maps\begin{align*}
&t^{R'P_1'\rightarrow P_2P_3}:=\sum_{i,s_1,s_2,s_3}t_{is_1s_2s_3}\ket{s_2s_3}^{P_2P_3}\bra{is_1}^{R'P_1'},\\
    &t^{R'P_2'\rightarrow P_1P_3}:=\sum_{i,s_1,s_2,s_3}t_{is_1s_2s_3}\ket{s_1s_3}^{P_1P_3}\bra{is_2}^{R'P_2'}\text{, and}\\
    &t^{R'P_3'\rightarrow P_1P_2}:=\sum_{i,s_1,s_2,s_3}t_{is_1s_2s_3}\ket{s_1s_2}^{P_1P_2}\bra{is_3}^{R'P_3'}.
\end{align*}
where $R',P_a'$ are copies of $R,P_a$. 
 %Instead, we only assume that their operator norms are small, and we will see that it implies almost minimal $I_3$ for all density matrices in the code space.
  \begin{lemma}\label{tensor3} For all $\trho^{P_1P_2P_3}\in\S$,
 \[ D(\rho^{RP_1}||\rho^{R}\ten \rho^{P_1})
 \kl 2\log \|t^{R'P_1'\to  P_2P_3}\| \pll .\]
\end{lemma}

\begin{proof}
By definition of the code space, we still have \begin{align*}
    \ti^{P_1P_2P_3}=id^{P_a}\otimes t^{R'P_1'\to  P_2P_3}\sum_{s_1}\frac{1}{\sqrt{d}}\ket{i}^{R'}\ket{s_1s_1}^{P_1'P_1}\pll .
\end{align*}
  Let $\trho=\sum_{i}\lambda_i\op{\wt{\psi_i}}{\wt{\psi_i}}$ be the spectral decomposition, then by superposition,\begin{align*}
    \ket{\wt{\psi_i}}^{P_1P_2P_3}=id^{P_a}\otimes t^{R'P_1'\to  P_2P_3}\sum_{s_1}\frac{1}{\sqrt{d}}\ket{\psi_i}^{R'}\ket{s_1s_1}^{P_1'P_1}\pll .
\end{align*}
So the purification satisfies\begin{align*}
    \ket{\phi}^{RP_1P_2P_3}:&=\sum_i\sqrt{\lambda_i}\ket{\psi_i}^R\ket{\wt{\psi_i}}^{P_1P_2P_3}\\
    &=id^{P_a}\otimes t^{R'P_1'\to  P_2P_3}\left(
    \sum_{i,s_1}\sqrt{\lambda_i/d}\ket{\psi_i}^{R'}\ket{s_1}^{P_1'}\otimes
    \ket{\psi_i}^{R}\ket{s_1}^{P_1}
    \right)
\end{align*}Let $\ket{\chi}$ be the vector on the right hand side before applying $t$. Then we get
 \begin{align*}
   \rho^{RP_1} &=  tr_{P_2P_3}(\op{\phi}{\phi}) \\
  &= tr_{P_2P_3}\left(t\op{\chi}{\chi} t^\dagger \right) \\
  &\le \|t\|^2 tr_{R'P_1'}(\op{\chi}{\chi}) \\
  &= \|t\|^2 \sum_{i} \lambda_i \op{\psi_i}{\psi_i} \ten \frac{\id}{d}  \\
  &= \|t\|^2 \rho^R\ten \rho^{P_1}
    \end{align*}
Recall that (see \cite{Wilde})
 \[ D(\rho||\si) \kl  D_{\infty}(\rho||\si)
 \lel \inf \{\la| \rho \kl 2^{\la} \si \} \pll .\]
In particular,
 \[ D(\rho^{RP_1}|| \rho^R\ten \rho^{P_1})
 \kl \log \|t\|^2 \pll .\]
The assertion follows.
\end{proof}\begin{cor}  Let $\S\subset P_1P_2P_3$ be a coding subspace  and $t$ the tensor as above. Then
 \begin{align*}
  -2S(R)&\le  I_3(P_1:P_2:P_3) \kl  -2S(R)\\
  &\quad  +
 2\log\|t^{R'P_1'\to P_2P_3}\|+
 2\log\|t^{R'P_2'\to P_1P_3}\|+
 2\log\|t^{R'P_3'\to P_1P_2}\| \pll ,
 \end{align*}
 for any $\trho\in\S$.
\end{cor}

\begin{proof} We have seen above that
 \begin{align*}
  I_3(R:P_1:P_2)\lel -2S(R)+I(R,P_1)+I(R,P_2)+I(R,P_3) \pll .
  \end{align*}
Therefore applying Lemma \ref{tensor3} three times, we get the assertion.
\qd

  %%%%%%%%%%%%%%%%%%%

\section{Random estimates}

Page scrambling \cite{Page} turned out to be of fundamental importance in building up a suitable theory for black holes (see \cite{Eng1,Eng2,Eng3}). We will briefly indicate how random unitaries for $RP_1,P_2P_3$ will deliver low $I_3$ estimates. Here we will assume that $H=RP_1$ and $|R|=|P_1|=|P_2|=|P_3|$.  Let us start with some probabilistic background by fixing a basis $(e_t)_{1\le t\le d^2}$ for $H$. Recall that a random unitary here means Haar distributed. A complex gaussian matrix is of the form
\[ g \lel (g_{st})_{s,t} \]
such that each entry
 \[ g_{st}\lel \frac{g_{st}(1)+ig'_{st}(2)}{\sqrt{2}} \]
is given by a complex gaussian entry, which are i.i.d. The following Lemma is implicitly contained in \cite{MP}. Indeed, we use the complex version of \cite[Corollary 2.4]{MP} by considering the Banach space $X$ obtained from $\Mz_n$ equipped with the semi-norm. Then we may write
 \[ g\lel \sum_{rs} tr(g|s\ran\lan r|)|r\ran\lan s| ) \pll .\]
For a reader interested in seeing how \cite{MP} directly implies our next result, we suggest to work with the Banach space valued matrices  $x_{rs}=|s\ran\lan r|\ten |r\ran\lan s| \in \Mz_n(X)$. For the convenience of the reader and an explicit control of constants, we provide the proof for a special case of \cite{MP}.

\begin{theorem} Let $\|\pll\|$ be a semi-norm on $\Mz_n$ and $1\le p\le \infty$. Then
 \[  \frac{\sqrt{n}}{8} (\ez \|u\|^p)^{1/p} \kl (\ez \|g\|^p)^{1/p} \kl 4\sqrt{pn} (\ez \|u\|^p)^{1/p} \pll .
  \]
\end{theorem}

\begin{proof} We will make frequent use of the Khintchine-Kahane inequalities (which has bet constant $\sqrt{2}$ for comparing $\|\pll \|_2$ and $\|\pll\|_1$ norms see \cite{Latala}), and Chevet's inequality see e.g. \cite{LT}.
Let $g=(g_{rs})$ be the complex gaussian variable from above. For complex unitaries $u,w$ we see that $ugw=_{D}g$ has the same distribution.  For a matrix $g$ we recall that $g=u_gD_{s(g)}w_g$ is the singular value decomposition with diagonal matrix $D_{s(g)}$ given by the singular values, i.e. the eigenvalues of the absolute value, $s_j(g)=\la_j(|g|)$ . Therefore, we deduce equality in distribution
 \[ g =_{D} ugw =_{D} u u_{g}D_{s(g)}w_gw  =_{D} u D_{s(g)}w =_{D} uM_{\si}D_{s(g)}M_{\si^{-1}}w
\pll . \]
Here $M_{\si}$ is a permutation matrix.
This implies by convexity
 \begin{align*}
 (\ez \|g\|^p)^{1/p} &= (\ez_{u,w,\si} \|u M_{\si}D_{s(g)}M_{\si^{-1}} w\|^p) \\
 &\gl (\ez_{u,w} \|u \ez_{\si} (M_{\si}D_{s(g)}M_{\si}^{-1}) w\|^p)^{1/p} \\
 &= \ez \frac{1}{n} tr(|g|)
 (\ez_{u,w}  \|uw\|^p)^{1/p} \\
 &= \ez \frac{1}{n} tr(|g|) (\ez \|u\|^p)^{1/p} \pll .
 \end{align*}
An upper bound for $\|g\|_{\infty}$ follows from Chevet's inequality for real gaussian matrices (see \cite{LT}), namely
 \begin{align*}
  (\ez \|g^{\cz}\|_{\infty}^p)^{1/p}&\le
  \sqrt{2p} \ez \|g^{\cz}\|_{\infty} \\
 &\kl 2\sqrt{p} \ez \|g^{\rz}\|_{\infty} \\
 &\kl 4\sqrt{p} \sqrt{n} \pll .
 \end{align*}
Thus by duality, we deduce from the Khintchine-Kahane inequality that
 \begin{align*}
  n^2 &= \ez tr(g^*g)  \kl \ez \|g\|_{\infty} \|g\|_1 \\
  &\kl (\ez \|g\|_1^2)^{1/2} (\ez \|g\|_{\infty}^2)^{1/2} \\
  &\kl 4\sqrt{2} \sqrt{2} \sqrt{n} \ez \|g\|_1 \pll .
\end{align*}
Here $\|g\|_1$ and $\|g\|_{\infty}$ refer to the trace class, or operator norm, respectively. This completes the proof of the lower estimate. For the upper estimate, we may use an extreme point argument, i.e., for any diagonal matrix $D_{\la}$, it can be written as
 \[ D_{\la} \lel  \summ_{\eps} \al(\eps,\la) D_{\eps} \]
such that
 \[ \sum_{\eps} |\al(\eps,\la)| \kl \|\la\|_{\infty} \pll .\]
Indeed, the extreme points of $[-1,1]^n$ are given by the $\om 1$ matrices $\eps$. This allows us to use the triangle inequality for fixed $g$ and random $u,w$
 \begin{align*}
 (\ez_{u,w} \|uD_{\si(g)}w\|^p)^{1/p}
 &= (\ez_{u,w} \|u\sum_{\eps} \al(\eps,g)D_{\eps}w\|^p)^{1/p} \\
 &\le \sum_{\eps} |\al(\eps,g)|
 (\ez \|uD_{\eps}w\|^p)^{1/p} \\
 &\le \|g\|_{\infty} (\ez \|u\|^p)^{1/p} \pll .
 \end{align*}
Integrating this over $g$ implies that
 \[ (\ez \|g\|^p)^{1/p}
 \kl (\ez \|g\|_{\infty}^p)^{1/p} (\ez \|u\|^p)^{1/p} \pll .\]
The assertion follows. \qd

\begin{cor} Let $|R|=d$ and
$u:RP_1\to P_2P_3$ be a random unitary. Let $\delta>0$.  Then with probability $1-\delta$
 \[ \max\{\|u\|,\|u^{R}\|,\|u^{\Gamma}\|\} \kl 48e\sqrt{2 \log \frac{1}{\delta}} \pll .\]
where $u^R, u^{\Gamma}$ are defined as the reshuffling and partial transpose with a flip, defined in Theorem \ref{lemma:m-u}.
\end{cor}

\begin{proof} Let $g=(g_{ab})$ be a complex gaussian matrix. Since $H=SP_1$ we may assume that $a=(i,s)$ is given by pairs and $g_{ab}$ is given by $m=d^2$ many complex independent gaussian random variables.  Note that the map $w(e_{i,s,j,r})=e_{i,r,j,s}$ is a permutation unitary and hence preserves the norm:
  \[ (\ez \|(w(g)_{ab})\|_{\infty}^p)^{1/p} \lel  (\ez \|(g_{ab})\|_{\infty}^p)^{1/p}
  \kl 2\sqrt{2p} \sqrt{m} \pll .\]
The same applies for all permutation of the indices.  Let us introduce the new seminorm
 \[ \|g\|\lel \max\{\|g\|_{\infty},\|g^{R}\|_{\infty},\|g^{\Gamma}\|_{\infty}\} \pll .\]
Then we deduce  from the triangle inequality that
 \begin{align*}
 (\ez \|u\|^p)^{1/p}
 &\kl \frac{8}{\sqrt{m}} (\ez \|g\|^p)^{1/p} \\
 &\kl \frac{24}{\sqrt{m}} (\ez \|g\|_{\infty}^p)^{1/p} \\
 &\kl  48 \sqrt{2p}   \pll .
 \end{align*}
Thus for every $\la\gl 1$, by Chebyshev inequality, we see that
 \begin{align*}
 {\rm Prob}(\|u\|\gl \la) &\kl  48^p \la^{-p} \sqrt{2p}^p \pll .
 \end{align*}
Now we choose $p=(\frac{\la}{48e\sqrt{2}})^2$ and deduce
   \[ {\rm Prob}(\|u\|\gl \la)\kl  e^{-p} = e^{-(\frac{\la}{48e\sqrt{2}})^2} \pll .\]
Thus it suffices to choose $\la= 48e\sqrt{2 \log \frac{1}{\delta}}$ so that $p\geq 1$. \qd

Combining these estimates, we obtain the following result.

\begin{theorem} Let $d=|R|=|P_1|=|P_2|=|P_3|$ and $\delta>0$. Let $u:RP_1\to P_2P_3$ be a random unitary and
 \[ |\tilde{i}\ran \lel \frac{1}{\sqrt{d}}
 \sum_{s_1,s_2,s_3}
 u_{(i,s_1),(s_2,s_3)} |s_1s_2s_3\ran \pll .\]
With probability  $1-e^{-
\mu}$ the estimate
 \[ I_3(P_1:P_2:P_3)\kl -2S(R)+ \ln(48\sqrt{2}) + 6 +3\ln \mu
   \]
holds for all purified input states.
\end{theorem}

\begin{remark} Here we used the natural logarithm. With the base two, we find our estimate is worse than the estimates from \cite{Yoshida}, but it works for all the states. With high probability we just have to allow for a small number of additional bits.
\end{remark}  

%%%%%%%%%%%%%%%%%%

%\input{Extension/almost1}
\section{examples}
\subsection{A perfect secret sharing protocol for arbitrary dimension}\hfill\par
We first characterize a permutation code space $\S \subset P_1 P_2 P_3$ by fixing its basis to be
\begin{align}\label{eqn:permu}
  \ket{ \tilde{i}}^{P_1P_2P_3} = \frac{1}{\sqrt{n}}\sum_{s=1}^{d}
    \ket{\sigma_1^i(s)}^{P_1}\ket{\sigma_2^i(s)}^{P_2}\ket{\sigma_3^i(s)}^{P_3},
    \ i = 1,\cdots,d,
\end{align}
where $\sigma_{j}$ denotes a permutation operator in $S_d$ for $j \in \{1,2,3\}$,
and $\sigma_{j}^i$ denotes the composition of $\sigma_j$ for $i$ times.
\begin{prop}
  The code space $\S$ is a ((2,3)) threshold scheme if and only if for any $a\neq b\in\{1,2,3\}$, \begin{equation}
  \label{ortho2}
\sigma_a^i(s) \neq \sigma_b^i(s)\ \text{for all}\  s \in \{1,\cdots, d\}, i \in \{1,\cdots, d-1\}.
\end{equation}
\end{prop}

\begin{proof}
From Definition \ref{def:secretsharing}, we use the definition (a) of a secret sharing scheme. Namely, we need to show that the condition (\ref{ortho2}) is equivalent to $I_3=-2S(R)=-2\log d$ for the state $\trho=\sum_i\op{\wt i}{\wt i}/d$.\par
It is easy to show that for any $j\in\{1,2,3\}$,
$$\trho^{P_j}=\frac{1}{d}\sum_s \op{s}{s}^{P_j},$$
 and therefore we have $S(P_j) = \log d$. Moreover, for any $a\neq b \in \{1,2,3\}$,
$$
\trho^{P_aP_b}=\frac{1}{d^2}\sum_{is}\op{\sigma_a^i(s)\sigma_b^i(s)}{\sigma_a^i(s)\sigma_b^i(s)}^{P_aP_b},
$$
so $S(P_aP_b) \leq 2\log d$, with equality holds if and only if
  $\{\ket{\sigma_a^i(s)\sigma_b^i(s)}^{P_aP_b}\}_{i,s = 1}^{d}$ forms an orthonormal basis. Note that the orthonormality requirement is equivalent to the condition (\ref{ortho2}). In addition,
  \begin{align*}
      I_3&=\sum_j S(P_j)-\sum_{a\neq b}S(P_aP_b)+S(R)\\
      &\geq 3\log d - 3\cdot 2\log d +\log d =-2\log d,
  \end{align*}
  with equality obtained if and only if equality hold for $S(P_aP_b) \leq 2\log d$. Thus condition (\ref{ortho2}) is satisfied if and only if $\S$ is a ((2,3)) threshold scheme.
\end{proof}
Using a concrete set of permutations that satisfies (\ref{ortho2}), we now provide a ready-to-use ((2,3)) threshold scheme. Remark that this protocol works for all dimensions $d$, which is an improvement over the existing examples of minimal $I_3$.
\begin{example}\label{concrete} Let $\S\subset P_1P_2P_3$ be generated by
  \begin{equation}
    \label{the-ex}
    \ket{\wt{i}} = \frac{1}{\sqrt{d}} \sum_{s=1}^d \ket{s}^{P_1}\ket{s + k_1 i}^{P_2}\ket{s + k_2 i}^{P_3},
  \end{equation}
  where $k_1\neq k_2$, and both $k_1$ and $k_2$ are coprime with $d$. The additions are mod $d$.
  \end{example}
  It is not hard to verify that this is indeed a permutation code space in the form of \eqref{eqn:permu}, and that it satisfies \eqref{ortho2}. Therefore, the code space is a ((2,3)) threshold scheme. \par
  We can also find the unitaries for erasure correction in the definition (b) of secret sharing scheme. We write the basis in the tensor form\begin{align*}
      \ket{\wt{i}} = \frac{1}{\sqrt{d}} \sum_{s_1s_2s_3}t_{is_1s_2s_3}\ket{s_1s_2s_3}^{P_1P_2P_3},
  \end{align*}
  where $t_{is_1s_2s_3}=\delta_{s_2,s_1+k_1i}\delta_{s_3,s_1+k_2i}$. Then we have the usual construction for the error correcting unitaries using
  \eqref{eqn:get-u}.

\par Note that this is a generalization of a well-known example \cite{Harlow, Lo} of ((2,3)) threshold
scheme given by
\begin{align*}
& \ket{\widetilde{0}} = \frac{1}{\sqrt{3}}(\ket{000}+ \ket{111} + \ket{222}), \\
& \ket{\widetilde{1}} = \frac{1}{\sqrt{3}}(\ket{012}+ \ket{120} + \ket{201}), \\
& \ket{\widetilde{2}} = \frac{1}{\sqrt{3}}(\ket{021}+ \ket{102} + \ket{201}).
\end{align*}
Our example \ref{concrete} reduces to this when $d=3$, $k_1=1$, $k_2=2$.

\subsection{An imperfect secret sharing protocol with a VIP party}\hfill
\par We provide an secret sharing protocol such that after the referee send a secret to $P_1$, $P_2$ and $P_3$, \begin{itemize}
    \item $\{P_1,P_3\}$ or $\{P_2,P_3\}$ together can reconstruct the secret, but
    \item $\{P_1,P_2\}$ together cannot reconstruct the secret.
\end{itemize}
It is as if the party $P_3$ is a VIP, since in order to reconstruct the secret, party $P_3$ has to be present. However, $P_3$ is not too powerful because he alone still cannot decode the message.
\par  We define the code space $\S\subset P_1P_2P_3$ by fixing the basis
\begin{equation}
  \ket{\widetilde{i}} = \sum_{j,k,l = 1}^{d}  \frac{1}{\sqrt{d}} t_{ijkl}\ket{jkl}^{P_1P_2P_3}\text{, where }
  t_{ijkl} := \frac{1}{\sqrt{d}}\bra{j} \lambda_k u_l \ket{i},
\end{equation}
where $\lambda_k$ is a shift operator and $u_{l}$ is a phase shift operator such that
\begin{equation}
  \lambda_k: \ket{j} \mapsto \ket{j+k},\q\q  u_l: \ket{i} \mapsto w^{il}\ket{i},
\end{equation}
where $w=e^{2\pi i/d}$. One can define the maps as usual,
\begin{align*}
    &t^{R'P_1'\to P_2P_3}:=\sum t_{ijkl}\op{kl}{ij},\\
    &t^{R'P_2'\to P_1P_3}:=\sum t_{ijkl}\op{jl}{ik},\\
    &t^{R'P_3'\to P_1P_2}:=\sum t_{ijkl}\op{jk}{il}.
\end{align*}
We can verify that indeed $t^{R'P_1'\to P_2P_3}$ and $t^{R'P_2'\to P_1P_3}$ are unitaries but not $t^{R'P_3'\to P_1P_2}$, thus giving the pairs $\{P_1,P_3\}$ and $\{P_2,P_3\}$ the ability to recover the secret, but not $\{P_1,P_2\}$.\par
We show the calculation for $t^{R'P_1'\to P_2P_3}$ as an example:
\begin{align*}
\left(t^{R'P_1'\to P_2P_3}\right)
\left(t^{R'P_1'\to P_2P_3}\right)^\dagger=
\sum_{k,k',l,l'}\left(
\sum_{ij}t_{ijkl}t_{ijk'l'}^*
\right)
\op{kl}{k'l'},
\end{align*}
where \begin{align*}
    \sum_{ij}t_{ijkl}t_{ijk'l'}^*&=\frac{1}{d}
    \sum_{ij}\bra{i}u_{l'}^\dagger \lambda_{k'}^\dagger \ket{j}
    \bra{j}\lambda_k u_l\ket{i}\\
    &=\frac{1}{d}\Tr\left(u_{l'}^\dagger u_l\right)\delta_{kk'}=\delta_{ll'}\delta_{kk'}.
\end{align*}
One can also verify that $\sum_{ik}t_{ijkl}t_{ij'kl'}^*=\delta_{jj'}\delta_{ll'}$. Thus the maps $t^{R'P_1'\to P_2P_3}$ and $t^{R'P_2'\to P_1P_3}$ are unitary. But we have $\sum_{il}t_{ijkl}t_{ij'k'l}^*=\delta_{l'-j',l-j}$, so the map $t^{R'P_3'\to P_1P_2}$ is not unitary. \par
From Theorem \ref{lemma:m-u}, if there were a decoding scheme for $\{P_1,P_2\}$, the error-correcting unitary must be uniquely defined to be equal to $t^{R'P_3'\to P_1P_2}$. But here we do not have the unitarity, so there is no decoding scheme for parties $\{P_1,P_2\}$.\par
Moreover, we must note that this is not the trivial case where all the secret is contained in $P_3$. It can be shown that $$\| t^{R'P_3'\to P_1P_2}\|^2=d.$$
Thus from Lemma \ref{tensor3}, $I(R,P_3)\leq \log d$ for any $\trho\in\S$. So at least we can show for the maximally mixed state $\sum_i\op{\wt i}{\wt i}/d$,$$
I(R,P_3)\leq \log d<2\log d=I(R,P_1P_2P_3),
$$which implies that party $P_3$ alone cannot recover the secret. \par
Moreover, we have \begin{equation}
  -2S(R) \leq I_3(\trho) \leq -2S(R) + \log d\text{, for any }\trho\in\S.
\end{equation}
Interestingly, we see that $I_3$ remains non-positive for both the pure state ($I_3=0$) and the maximally mixed state ($I_3=-\log d$). Our conjecture is that $I_3\leq 0$ holds for all $\trho\in\S$. This property is called monogamy and has significant implications in the context of holography and AdS/CFT correspondence \cite{Hayden_2013}.
%%%%%%%%

%\input{Example/example}

\section{Conclusion}
In summary, our note develops a connection between tripartite information $I_3$ and secret sharing protocols. In particular, we observed that the sharing protocol is perfect if and only if the tripartite information is minimal for all states in the secret sharing protocol. Moreover, we showed that perfect secret sharing protocol is also equivalent to the recovery unitary defined in Harlow coming from multi-unitary. \\
\par Based on the connection of tripartite information and perfect secret sharing protocol, we find imperfect sharing schemes given by Page-scrambling unitaries working for almost all of Alice's secrets and VIP models with preference to one fo three parties.
%Based on the connection of tripartite information and perfect secret sharing protocol, we talk about some extensions. First we give an example where only subset of parties can recover the secret but not perfect for all parties. Moreover, for general secret sharing protocal, which may not be perfect for any parties, we have the estimate of tripartite information, given by the operator norm of the tensor map. As an application, for random unitaries, we give the random estimates of $I_3$, with a dimension free upper bound for any the densities in the secret sharing protocol.
%\input{Conclusion/conclusion}

%\bibliographystyle{unsrt}
\bibliographystyle{IEEEtran}

\bibliography{}

\end{document}